\documentclass[pdflatex,sn-mathphys-num]{sn-jnl}

\usepackage{graphicx}%
\usepackage{multirow}%
\usepackage{amsmath,amssymb,amsfonts}%
\usepackage{amsthm}%
\usepackage{mathrsfs}%
\usepackage[title]{appendix}%
\usepackage{xcolor}%
\usepackage{textcomp}%
\usepackage{manyfoot}%
\usepackage{booktabs}%
\usepackage{algorithm}%
\usepackage{algorithmicx}%
\usepackage{algpseudocode}%
\usepackage{listings}%
\usepackage{longtable}

\usepackage{diagbox}
\usepackage{float}
\usepackage{tikz}
\usepackage{caption}
\usepackage{hyperref}

\theoremstyle{thmstyleone}%
\newtheorem{theorem}{Theorem}
\newtheorem{proposition}[theorem]{Proposition}%

\theoremstyle{thmstyletwo}%

\theoremstyle{thmstylethree}%
\newtheorem{definition}{Definition}

\newtheorem{notation}{Notation}
\newtheorem{corollary}{Corollary}

\begin{document}

\title[Article Title]{Hardy's Paradox for Yu-Oh Set Constructed by Logically Contextual Quantum States}


\author[1]{\fnm{Chang} \sur{He}}\email{hechang@buaa.edu.cn}

\author*[1]{\fnm{Yongjun} \sur{Wang}}\email{wangyj@buaa.edu.cn}

\author[1]{\fnm{Baoshan} \sur{Wang}}\email{bwang@buaa.edu.cn}

\author[1]{\fnm{Songyi} \sur{Liu}}\email{liusongyi@buaa.edu.cn}

\author[1]{\fnm{Yunyi} \sur{Jia}}\email{by2309005@buaa.edu.cn}

\affil[1]{\orgdiv{School of Mathematical Scieneces}, \orgname{Beihang University}, \orgaddress{ \city{Beijing}, \postcode{100191}, \country{China}}}


\abstract{
	Quantum contextuality is a fundamental nonclassical property of quantum systems, regarded as a key resource that demonstrates the computational and informational advantages of quantum over classical systems. Our present work aims to construct Hardy's paradoxes, a set of possibilistic conditions witnessing contextuality, for Yu-Oh set, which is the state-independent contextual quantum system with the least number of vectors. To achieve the aim, we systematically enumerate all logically contextual pure states on Yu-Oh set, and theoretically prove that no mixed states in this scenario are logically contextual. Based on the identified logically contextual quantum states, we construct 12 Hardy's paradoxes with identical success probability SP=11.1\%. Furthermore, we present corresponding observables to experimentally witness these Hardy's paradoxes.}


\maketitle

\section*{Declarations}

\noindent\textbf{Competing interests.}
The authors have no relevant financial or non-financial interests to disclose.

\vspace{10pt}

\noindent\textbf{Acknowledgments.}
The work was supported by National Natural Science Foundation of China (Grant No. 12371016, 11871083) and National Key R\&D Program of China (Grant No. 2020YFE0204200).

\section{Introduction}\label{sec1}
 
Quantum contextuality is referred to as a central result in the foundations of quantum mechanics \cite{ksgeneral2022}, because it conflicts with the classical $deterministic\ local\ hidden-variable\ model$, so as to provide an answer to the Einstein-Podolsky-Rosen (EPR) argument \cite{epr1935}. The earliest study of quantum contextuality has started since the introduction of Bell's non-locality theory \cite{bell1964} in 1964. The Bell non-locality is technically contextuality with entanglement on a space-like separated scenario and it has been experimentally verified in numerous setups, including the Clauser-Horne-Shimony-Holt (CHSH) experiment \cite{chsh1969} and the Klyachko-Can-Binicioglu-Shumovsky (KCBS) experiment \cite{kcbs2008}. In 1967, the Kochen-Specher theorem \cite{ks1967} pioneered a field called quantum contextuality and its refinement and relaxations have been proposed in many theoretic structures such as sheaf-theoretic approach \cite{ab2011}, graph-theoretic approaches \cite{cabello2014,silva2017,liu2025} and quantum logic \cite{ab2021,Zhou2024,liu2025}. Since the introduction of non-locality and contextuality by Bell, Kochen and Specker in the 1960s, this field has gained significant theoretical and applied developments. Over the years, contextuality has been recognized as an vital resource for numerous applications in quantum information \cite{Jon2005,Jon2007,Buhrman2010,kar2016}, quantum communication \cite{gupta2023, bennett1992, barrett2005}, and quantum cryptography \cite{ekert1991, gisin2002}.

Hardy's argument \cite{Hardy1992,Hardy1993} reveals contextuality in bipartite Bell-scenario for two-qubit partially entangled states simply using a set of possibilistic conditions, one of which characterizes an event with probability greater than 0 and others 0. The possibilistic conditions later leads to a concept of Hardy's paradox. As the only $non-zero$ condition starts the chain of logical reasoning for Hardy's argument and it depicts the probability to witness a certain Hardy's paradox in an experiment, the probability of this important condition is referred to as the success probability $SP$.

The conciseness of Hardy's paradox has drawn much attention to make it cater for more generalized quantum systems. In the present, Hardy's paradox has been found to cater for $(2, k, 2)$ \cite{2k21992} and $(n, 2, 2)$ \cite{n221992} Bell-scenario, and for any partially entangled 2-qubit states \cite{Hardy1993}. In 1996, it is proved by Kar \cite{Kar1997} that for a system of two spin-$\frac{1}{2}$ particles, no mixed state admits Hardy's paradox and there are mixed states for a system of three spin-$\frac{1}{2}$ particles which show Hardy's paradox. From Kar's idea, Cabello \cite{Cabello2002} and later Liang and Li \cite{Liang2003,Liang2005} proposed a relaxation of Hardy's paradox. The former gives a Hardy-type logical proof of contextuality for one class of 3-qubit mixed states, and the latter 2-qubit. Their constructions relax one of the $zero$ conditions of the paradox to a $non-zero$. This relaxation allows them to demonstrate the generalized paradox with degree of success $DS$, which was first proposed in \cite{Rai2021} as a generalized notion of success probability $SP$. Over the years, investigation of higher $SP$ and $DS$ has been introduced in succession \cite{ds2013,ds2017,ds2018}. While researches on Hardy's paradox based on generalized Bell scenario have achieved considerable success, we manage to extend the construction of Hardy's paradox for Yu-Oh set and give its corresponding $SP$.

Yu-Oh set \cite{yuoh2012} is a state-independent contextual (SIC) quantum system with 13 vectors, shown in Figure~\ref{fig1}. Each vector corresponds to a rank-1 projector. It is proved to be the simplest SIC scenario with the least number of rank-1 projectors by Cabello \cite{cabello2016}. Yet due to its projector-based definition, Yu-Oh set has not been taken into account in many measurement-based theoretic frameworks. It was pointed out in \cite{ksgeneral2022} that the classifications given in sheaf-theoretic framework can be considered in some sense incomplete as they take into account only Bell-models, and thus fail to recognize the importance of contextuality arguments such as the one by Sixia Yu and C.~H.~Oh \cite{yuoh2012}. Constructing Hardy's paradox for Yu-Oh set is not only an extension of Hardy's argument but a step to depict Yu-Oh set in the existing theoretic frameworks for contextuality researches.

In 1989, the first logical contradiction revealing nonlocality, on a four-qubit maximally entangled state \cite{GHZ1989}, three-qubit later in 1990 \cite{GHZ1990}, was first established by Greenberger, Horne, and Zeilinger. Subsequently, a different type of logical contradiction was provided in Hardy's argument \cite{Hardy1992,Hardy1993}. These logical proofs of contextuality lead to a hierarchy of contextuality introduced by Abramsky and Brandenburger in 2011 \cite{ab2011}. Later, Silva established a refined version of this hierarchy with graph-theoretic approach \cite{silva2017}. In the hierarchy of contextuality, Hardy's argument leads to a concept called $logical\ contextuality$, also referred to as $possibilistic\ contextuality$. Logical contextuality characterizes systems for which no global distribution in the Boolean algebra $\mathbb{B}=\{0,1\}$ can be found for the distribution of marginal events in any contexts. Relationship between logical contextuality and Hardy's paradox has been found in many researches. Hardy's paradox has been proved to be a necessary and sufficient condition for logical contextuality in any $(2,2,l)$ or $(2,k,2)$ Bell scenario \cite{hp2012}, as well as in general $n$-cycle scenarios \cite{hp2021}. Although it is pointed out by Mansfield and Fritz \cite{mansfield2012} that coarse-grained Hardy paradox can not be constructed in some certain logically contextual scenarios, it is still a promising approach to construct Hardy's paradox with the help of logically contextual quantum states, especially during the construction for Yu-Oh set.

In order to reveal the logical properties of Yu-Oh set, we adopt an idea rooted in quantum logic pioneered by Birkhoff and Von Neumann \cite{birkhoff1936} in the 1930s. More specifically, we view the 13 rank-1 projectors in Yu-Oh set as quantum events, and thus Kochen-Specker-assignments of them induces global events for the scenario of Yu-Oh set. Then we can use Proposition~\ref{prop2} to enumerate all logically contextual quantum states. Based on these states, we finally manage to constructed 12 Hardy's paradoxes and give each of their $SP$s as 11.1\%.

In Section~\ref{sec2}, we provide an approach to fit Yu-Oh set in the hierarchy established in \cite{ab2011} and \cite{silva2017}. In Section~\ref{sec3}, we conduct a thorough identification of logically contextual quantum states in Yu–Oh quantum system. We enumerated all logically contextual pure states with Proposition~\ref{prop3} and Algorithm~\ref{alg1}. And theoretically proved the non-existence of logically contextual mixed states with Propositions \ref{prop5} to \ref{prop8}. In Section~\ref{sec4}, we demonstrate that each logically contextual quantum state gives rise to a Hardy's paradox for Yu–Oh set, supporting the claim that the occurrence of a Hardy's paradox is a necessary condition for logical contextuality. We also give the $SP$ of each Hardy's paradoxes as evaluation. Furthermore, in Section~\ref{sec5}, we propose corresponding observables to experimentally verify these Hardy's paradoxes.

\section{Quantum scenario and logical contextuality}\label{sec2}

\subsection{Quantum events and scenario}

Under the assumption of ideal measurements(i.e., the same measurement performed on the same observable in succession will only get the same result), observable $A$ corresponds to a bounded self-adjoint operator $\hat{A}$ on the \textit{Hilbert space} $\mathcal{H}$. $\hat{A}$ has unique spectral decomposition $\hat{A} =\sum _{i} a_i\hat{P_i}$, where $a_i$s are real and projectors $\hat{P_i}$s are mutually orthogonal. Then the probability of event $A=a_i$ is induced by quantum state $\rho$, denoted as $\rho(A=a_i)=tr(\rho\hat{P_i})$, where $\rho$ is a density operator in Hilbert space $\mathcal{H}$. So the event of measuring $A$ to get $a_{i}$ is equivalent to measuring projector $\hat{P_i}$ to get 1.

In this paper, we mainly focus on \textit{rank-1} projectors, each corresponds to a vector on the Hilbert space $\mathcal{H}$. For example, projector $\hat{P}=\frac{|v\rangle\langle v|}{\|v\|^2}$ corresponds to vector $|v\rangle$. We will hereafter refer to rank-1 projectors as vectors directly. In this sense, if a set of $d$ dimensional vectors $\mathcal{E}=\{v_i\}_{i=1}^n$\footnote{$v_i$ corresponds to vector $|v_i\rangle$.} is provided, then any orthogonal basis $\{v_j\}_{j=1}^d\subset\mathcal{E}$ corresponds to an observable $\hat{A}=\sum\limits_{j=1}^da_j\frac{|v_j\rangle\langle v_j|}{\|v_j\|^2}$, where $a_j$s are real numbers. And any orthogonal pair $v_m, v_n\in\mathcal{E}$, $\langle v_m|v_n\rangle=0$, corresponds to the different outcome of a single observable, which is a pair of exclusive events. So any observable event in the quantum system composed by $\mathcal{E}$ can be seen as truth assignment on an orthogonal basis $\{v_j\}_{i=1}^d\subset\mathcal{E}$, where one and only one vector $v_k\in\{v_j\}_{i=1}^d$ be assigned 1 and others 0. Then, without loss of generality, definition of a $quantum\ event$ can be generated as follows.

\begin{definition}
	A $quantum\ event$ in a quantum system is a non-zero vector $v$ on the Hilbert space $\mathcal{H}$. Orthogonal quantum events are exclusive events.
\end{definition}

In a quantum system where there are quantum events $\mathcal{E}$, any maximal pairwise-orthogonal subset $C\subset\mathcal{E}$\footnote{$C\subset \mathcal{E}$ satisfies: (\romannumeral1) $\langle v_i|v_j\rangle=0$ for all $v_i,v_j\in C$; (\romannumeral2)$\langle u|v\rangle\neq0$ for all $u\in\mathcal{E}\setminus C$ and $v\in C$} forms a context composed of a single observable $\hat{A}=\sum\limits_{v_i\in(C\cup C^\perp)}a_i\frac{|v_i\rangle\langle v_i|}{\|v_i\|^2}$, where $C^\perp$ is the orthogonal complete of $C$ and $a_i$s are real. A quantum scenario specifies the allowed observable events in a family of contexts, which means $\mathcal{E}$ itself characterizes a quantum scenario.

\begin{definition}
	A $quantum\ scenario$ is a set of quantum events $\mathcal{E}=\{v_i\}$. A $context$ in $\mathcal{E}$ is a maximal pairwise-orthogonal subset to $\mathcal{E}$, denoted as  $C\subset\mathcal{E}$. The set of all contexts in $\mathcal{E}$ is denoted as $\mathcal{C}(\mathcal{E})$.
\end{definition}

The scenario of Yu-Oh set can be seen as 13 quantum events, hereafter denoted as $\mathcal{E}_{Yu-Oh}$. The exclusivity graph of Yu-Oh set is provided as Fig~\ref{fig1}, in which events are represented by vertices, and pairs of exclusive events, whose corresponding projectors are mutually orthogonal, are represented by adjacent vertices. Any maximal clique of the graph corresponds to a context. For convenience, the vectors are provided omitted normalization, yet they should be normalized during calculation.

\begin{figure}[b]
	\centering
	\begin{tikzpicture}
		\coordinate (A) at (0, 4.5);
		\coordinate (B) at (2.6, 0);
		\coordinate (C) at (-2.6, 0);
		\coordinate (D) at (-0.6, 2.6);
		\coordinate (E) at (1.255, 1.47);
		\coordinate (F) at (-0.655, 0.43);
		\coordinate (G) at (0.6, 2.6);
		\coordinate (H) at (0.655, 0.43);
		\coordinate (I) at (-1.255, 1.47);
		\coordinate (J) at (-0.695, 1.9);
		\coordinate (K) at (0.695, 1.9);
		\coordinate (L) at (0, 0.7);
		\coordinate (M) at (0, 1.5);
		
		\draw (A) -- (B) -- (C) -- (A);
		\draw (A) -- (D) -- (G) -- (A);
		\draw (B) -- (E) -- (H) -- (B);
		\draw (C) -- (I) -- (F) -- (C);
		\draw (D) -- (J) -- (I);
		\draw (G) -- (K) -- (E);
		\draw (H) -- (L) -- (F);
		\draw (J) -- (H);
		\draw (K) -- (I);
		\draw (L) -- (G);
		\draw (M) -- (D);
		\draw (M) -- (E);
		\draw (M) -- (F);
		\draw[dashed] (M) arc (180:100:2);
		
		\node[above] at (A) {$v_1=(1,0,0)$};
		\node[below right] at (B) {$v_2=(0,1,0)$};
		\node[below left] at (C) {$v_3=(0,0,1)$};
		\node[xshift=-1.7cm] at (D) {$v_4=(0,1,-1)$};
		\node[xshift=1.7cm] at (E) {$v_5=(1,0,-1)$};
		\node[xshift=-0.5cm, yshift=-0.7cm] at (F) {$v_6=(1,-1,0)$};
		\node[xshift=1.6cm] at (G) {$v_7=(0,1,1)$};
		\node[xshift=0.5cm, yshift=-0.7cm] at (H) {$v_8=(1,0,1)$};
		\node[xshift=-1.6cm] at (I) {$v_9=(1,1,0)$};
		\node[xshift=-2cm] at (J) {$v_A=(-1,1,1)$};
		\node[xshift=2cm] at (K) {$v_B=(1,-1,1)$};
		\node[yshift=-1.4cm] at (L) {$v_C=(1,1,-1)$}; 
		\node[xshift=2.7cm, yshift=2cm] at (M) {$v_D=(1,1,1)$};
		\foreach \point in {A,B,C,D,E,F,G,H,I,J,K,L,M} {
			\fill (\point) circle (2pt);
		}
	\end{tikzpicture}
	\caption{Exclusivity graph of Yu-Oh set}
	\label{fig1}
\end{figure}
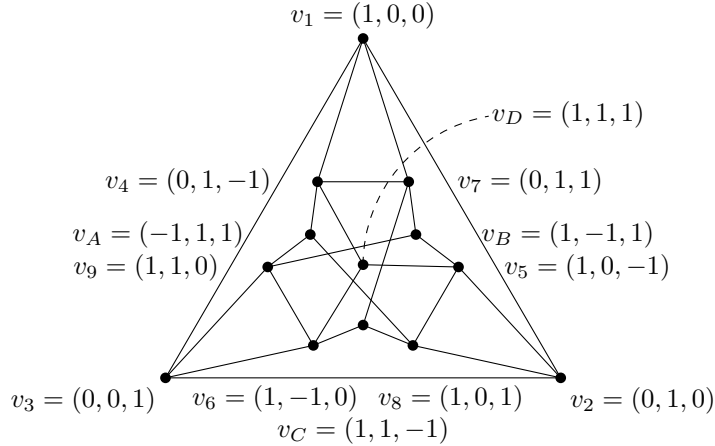

\subsection{Global events}
A structure called sheaf of events is proved in \cite{ab2011}, which allows the existence of a global section, which we can simply refer to as global events. It is also proved that each global event can be seen as a canonical form of \textit{deterministic hidden variable}. Then the nonexistence of probability distribution for these global events witnesses contextuality of the scenario.

In order to construct global events for a quantum scenario composed of vectors, we only need to identify all of its deterministic hidden variables. As generally acknowledged, a deterministic hidden variable assigns a definite outcome to each observable, independent of the context in which it appears. So it yields a 0-1 assignment for each observable context $C$, where one and only one vector $v_k\in C$ be assigned 1 and others 0. Such assignment was established by Kochen and Specker as follows.

\begin{definition}
	A \textit{Kochen-Specker-assignment}, abbreviated as \textit{KS-assignment}, for vectors $\mathcal{E}=\{v_i\}_{i \in I}$ is a function  $\lambda:\mathcal{E}\rightarrow\{0,1\}$ satisfying the following conditions.
	\begin{enumerate}
		\item[($O$)] Orthogonality: $\lambda(v_i)\lambda(v_j)=0$, if $\langle v_i|v_j\rangle=0$, $v_i,v_i\in\mathcal{E}$.
		\item[($C$)] Completeness: $\sum\limits_i\lambda(v_i)=1$, for all orthogonal basis $\{v_i\}\subset\mathcal{E}$.
	\end{enumerate}
\end{definition}

\begin{proposition}
	\label{prop1}
	In the scenario of Yu-Oh set $\mathcal{E}_{Yu-Oh}$, if function $\lambda_{KS}:\mathcal{E}_{Yu-Oh}\rightarrow\{0,1\}$ is a KS-assignment, then $\lambda_{KS}$ induces a deterministic hidden variable $\lambda$ on $\mathcal{E}_{Yu-Oh}$.
\end{proposition}

\begin{proof}
	Suppose $\lambda_{KS}:\mathcal{E}_{Yu-Oh}\rightarrow\{0,1\}$ is a KS assignment on $\mathcal{E}_{Yu-Oh}$. Let $C_k\in\mathcal{C}(\mathcal{E}_{Yu-Oh})$. According to Figure~\ref{fig1}, $C_k$ is either an orthogonal basis in the 3 dimensional Hilbert space, or a pair of orthogonal vectors.
	
	If $C_k=\{v_{k_1},v_{k_2},v_{k_3}\}$ is in an orthogonal basis, then there is one and only one vector $v_{k_i}$ be assigned 1 and others 0. Then $\lambda_{KS}$ induces a deterministic hidden variable $\lambda$ that assigns outcome $a_i$ to observable $A$, whose corresponding self-adjoint operator is $\hat{A}=\sum\limits_{i=1}^3a_i\frac{|v_{k_i}\rangle\langle v_{k_i}|}{\|v_{k_i}\|^2}$.
	
	If $C_k=\{v_{k_1},v_{k_2}\}$ is a pair of orthogonal vectors, on the one hand, if $\lambda_{KS}(v_{k_1})=\lambda_{KS}(v_{k_2})=0$ , then $\lambda_{KS}$ induces $\lambda$ that assigns outcome $0$ to observable $A$, whose corresponding self-adjoint operator is $\hat{A}=\sum\limits_{i=1}^2a_i\frac{|v_{k_i}\rangle\langle v_{k_i}|}{\|v_{k_i}\|^2}+0\cdot\frac{|v_{k_3}\rangle\langle v_{k_3}|}{\|v_{k_3}\|^2}$, where $v_{k_3}$ is orthogonal complete of $C_k$. On the other hand, suppose there is one vector $v_{k_i}$ be assigned 1 and others 0. Since none of context $C\in\mathcal{C}(\mathcal{E})$ of 2 vectors share the same orthogonal complete as calculated in Table A1 in Appendix~\ref{secA}, one and only one vector $v_{k_i}\in\{v_{k_i}\}_{i=1}^3$ can be assigned 1. Thus $\lambda_{KS}$ induces $\lambda$ that assigns outcome $a_i$ to observable $A$.
\end{proof}

Due to the fact that each deterministic hidden variable characterizes a global event for a certain scenario, so the KS-assignments for Yu-Oh set $\mathcal{E}_{Yu-Oh}$ actually induce the global events for it. The algorithm to enumerate all KS-assignments for a vector set is provided in Appendix~\ref{secA} as Algorithm A1, together with the table of all KS-assignments for Yu-Oh set as Table A2 in Appendix~\ref{secA}.

\begin{notation}
	To be more concise, we also use KS-assignments, denoted as $\lambda$, for a set of vectors $\mathcal{E}$, to characterize deterministic hidden variables, also global events, for $\mathcal{E}$. We denote the set of all $\lambda$ on $\mathcal{E}$ as $\Lambda(\mathcal{E})$, and $\Lambda(v)=\{\lambda|\lambda(v)=1\}$, where $v\in \mathcal{E}$.
\end{notation}

\subsection{logical contextuality}

\begin{definition}
	The \textit{possibiistic distribution} for global events $\Lambda(\mathcal{E})$ is function  $p_{\Lambda}:\Lambda(\mathcal{E})\rightarrow\mathbb{B}=\{0,1\}$ satisfying:
	\begin{equation*}
		\bigvee\limits_{\lambda\in\Lambda(\mathcal{E})}p_{\Lambda}(\lambda)=1
	\end{equation*}
	
	The \textit{marginal distribution} of $p_\Lambda$ for quantum event $v\in\mathcal{E}$ is:
	\begin{equation*}
		p_{\Lambda}|_C(v)=\bigvee\limits_{\lambda\in\Lambda(\mathcal{E}):\lambda(v)=1}p_{\Lambda}(\lambda),
	\end{equation*}
	where $C\in\mathcal{C}(\mathcal{E})$ and $v\in C$.
\end{definition}

In the classical probability theory, the marginal distribution of $p_{\Lambda}$ for quantum event $c$ must be equal to the possibilistic distribution $\overline{\rho}$ induced from quantum state $\rho$, where $\overline{\rho}(v)=\lceil \rho(v) \rceil=\lceil tr(\rho \frac{|v\rangle\langle v|}{\|v\|^2}) \rceil$. And the failure to find such possibilistic distribution $p_{\Lambda}$ for certain scenario witnesses logical contextuality. Thus the definition of logically non-contextual can be generated as follows.

\begin{definition}
	\label{def7}
	Quantum state $\rho$ is logically non-contextual on vector set $\mathcal{E}$, if and only if the following condition holds:
	
	There is a possibilistic distribution $p_{\Lambda}$ on $\Lambda(\mathcal{E})$ satisfying	
	\begin{equation}
		\label{eq1}
		\overline{\rho}(v)= p_{\Lambda}|_C(v),
	\end{equation}
	for all $v\in C$, $C\in\mathcal{M}(\mathcal{E})$.
\end{definition}

In \cite{ab2011}, the contradiction led by $p_{\Lambda}$ is found in a certain scenario by solving a linear system constructed with an incidence matrix. Yet this approach can only verify whether a quantum state is logically contextual on a certain scenario, and cannot enumerate such quantum states. A refinement of Definition~\ref{def7} is necessary.

\begin{proposition}
	\label{prop2}
	Quantum state $\rho$ is logically contextual on vector set $\mathcal{E}$, if and only if there exists $v\in\mathcal{E}$, $\overline{\rho}(v)=1$ such that for all global events $\lambda\in\Lambda(\mathcal{E})$ satisfying $\lambda(v)=1$, there exists $v'\in\mathcal{E}$,  $\overline{\rho}(v')=0$, where $v'\neq v$, $\lambda(v')=1$.
\end{proposition}

\begin{proof}
	On the one hand, if quantum state $\rho$ is logically contextual on vector set $\mathcal{E}$:
	
	Suppose for all $v\in\mathcal{E}$ that satisfies $\overline{\rho}(v)=1$, there exists a global event $\lambda\in\Lambda(\mathcal{E})$, $\lambda(v)=1$, such that $\overline{\rho}(v')=1$ for all $v'\in\mathcal{E}$, $\ v'\neq v$ and $\lambda(v')=1$, . Let:$p_{\Lambda}(\lambda)=\bigwedge\limits_{v:\lambda(v)=1}\overline{\rho}(v)$. Since $\bigvee\limits_{\lambda\in\Lambda(\mathcal{E})}p_{\Lambda}(\lambda)=\bigvee\limits_{\lambda\in\Lambda(\mathcal{E})}\bigwedge\limits_{v:\lambda(v)=1}\overline{\rho}(v)=1$, we know $p_{\Lambda}$ is a possibilistic distribution for $\lambda\in\Lambda(\mathcal{E})$.
	
	Then for all $v\in\mathcal{E}$ satisfying $\overline{\rho}(v)=1$:
	\begin{equation*}
		p_{\Lambda}|_C(v)=\bigvee\limits_{\lambda\in\Lambda(\mathcal{E}):\lambda(v)=1}p_{\Lambda}(\lambda)=\bigvee\limits_{\lambda\in\Lambda(\mathcal{E}):\lambda(v)=1}(\bigwedge\limits_{v:\lambda(v)=1}\overline{\rho}(v))=1.
	\end{equation*}
	Thus Equation~\ref{eq1} holds, which leads to contradiction.
	
	On the other hand, if there exists $v\in\mathcal{E}$, $\overline{\rho}(v)=1$, such that there exists $v'\in\mathcal{E}$, $v'\neq v$, $\lambda(v')=1$, $\overline{\rho}(v')=0$ for all global events $\lambda\in\Lambda(\mathcal{E})$, $\lambda(v)=1$:
	
	Suppose quantum state $\rho$ is logically non-contextual on vector set $\mathcal{E}$. Then there is a possibilistic distribution on $\Lambda(\mathcal{E})$, such that $\overline{\rho}(v)=p_{\Lambda}|_C(v)=1$. 
	
	In classical probability theory, we know if a global event is possible, none of its marginal event is impossible. Thus, $	p_{\Lambda}(\lambda)=\bigwedge\limits_{v':\lambda(v')=1}p_{\Lambda}|_{C'}(v')=\bigwedge\limits_{v':\lambda(v')=1}\overline{\rho}(v')$, where $v'\in C'$, $C'\in\mathcal{C}(\mathcal{E})$.
	
	Then,
	\begin{equation*}
		\begin{aligned}
			p_{\Lambda}|_C(v)=\bigvee\limits_{\lambda\in\Lambda(\mathcal{E}):\lambda(v)=1}p_{\Lambda}(\lambda)
			=\bigvee\limits_{\lambda\in\Lambda(\mathcal{E}):\lambda(v)=1}(\bigwedge\limits_{v':\lambda(v')=1}\overline{\rho}(v'))=0,
		\end{aligned}
	\end{equation*}
	which leads to contradiction.
\end{proof}

Proposition~\ref{prop2} provides a refinement for Definition~\ref{def7}, allowing us to specifically depict a quantum state $\rho$ being logically contextual on Yu-Oh set. A similar definition was proposed by Silva within the graph-theoretic approach \cite{silva2017}.

\section{Logically contextual quantum states on Yu-Oh set}\label{sec3}

\subsection{Pure states with logical contextuality on Yu-Oh set}\label{subsec3.1}

According to Proposition~\ref{prop2}, if quantum state $\rho$ is logically contextual on $\mathcal{E}_{Yu-Oh}$, there exists a vector $v_i\in\mathcal{E}_{Yu-Oh}$, $\rho(v_i)=tr(\rho|\frac{v_i\rangle\langle v_i|}{\|v_i\|^2})>0$ and for all global events $\lambda\in\Lambda(\mathcal{E}_{Yu-Oh})$, $\lambda(v_i)=1$, there is vector $v_j\neq v_i$, $ \lambda(v_j)=1$, $\rho(v_j)=0$. Thus we can discuss the selection of $v_i$ as a starting point.

We firstly categorize the 13 vectors depending on how much orthogonal basis $C\subset\mathcal{E}_{Yu-Oh}$ they are in. $V_1=\{v_1,v_2,v_3\}$, each of whose elements is in 2 orthogonal basis, $V_2=\{v_i\}_{i=4}^{9}$, each of whose elements is in 1, and $V_3=\{v_A,v_B,v_C,v_D\}$, each of whose elements is in none.

\begin{proposition}
	\label{prop3}
	For all vector $v\in V_1\cup V_2$, if $\rho(v)>0$, then there exists a global event $\lambda \in \Lambda(\mathcal{E}_{Yu-Oh})$, $\lambda(v)=1$, such that there is $\rho(v_i)>0$ for all quantum event $v_i\in\mathcal{E}_{Yu-Oh}$, $\lambda(v_i)=1$.
\end{proposition}

\begin{proof}
	If $v\in V_1$:
		
	Without loss of generality\footnote{This can be derived from the rotational symmetry property of Figure~\ref{fig1}, so is the reason for the similar argument in the second case.}, let $v=v_1$. Then there is $i\in \{5, 8\}$ such that $\rho(v_i)>0$. Otherwise, according to completeness($O$), $\rho(v_5)=\rho(v_8)=0\implies \rho(v_2)=1$, which contradicts that $\rho(s)=\rho(v_1)>0$. Without loss of generality, let $\rho(v_5)>0$. Similarly, there is $i\in \{6, 9\}$ such that $\rho(v_i)>0$. Without loss of generality, let $\rho(v_6)>0$. Till now, we have found a global event $\lambda$:
	$\lambda(v_i)=\begin{cases}
		1 & i\in \{1,5,6\}, \\
		0 & otherwise. \\
	\end{cases}$
	, satisfying for all marginal event $v_i$ such that $\lambda(v_i)=1$, $v_i\in\mathcal{E}_{Yu-Oh}$, there is $\rho(v_i)>0$.
		
	If $v\in V_2$:
		
	Without loss of generality, let $v=v_4$. Then there is $i\in \{2, 3\}$ such that $\rho(v_i)>0$. Otherwise, according to completeness($O$), $\rho(v_2)=\rho(v_3)=0\implies \rho(v_1)=1$, which contradicts that $\rho(s)=\rho(v_4)>0$. Without loss of generality, let $\rho(v_2)>0$.Then there is $i\in \{6, 9\}$ such that $\rho(v_i)>0$. Otherwise, according to completeness($O$), $\rho(v_6)=\rho(v_9)=0\implies \rho(v_3)=1$, which contradicts that $\rho(v_2)>0$. Without loss of generality, let $\rho(v_6)>0$.Till now, we have found a global event $\lambda$:
		$\lambda(v_i)=\begin{cases}
			1 & i\in \{2,4,6\}, \\
			0 & otherwise. \\
		\end{cases}$
		, satisfying for all marginal event $v_i$, $\lambda(v_i)=1$, there is $\rho(v_i)>0$.
	
	In summary, for all $s\in V_1\cup V_2$, if $\rho(s)>0$, then there exists $\lambda \in \Lambda(\mathcal{E}_{Yu-Oh}): \lambda(s)=1$ such that for all marginal event $v_i$, $\lambda(v_i)=1\ v_i\in\mathcal{E}_{Yu-Oh}$, there is $\rho(v_i)>0$.
\end{proof}

\begin{notation}
	To make it easier to express, we will use sets to characterize global events later. For example, we use $s_\lambda=\{v_1,v_5,v_6\}$ to characterize $\lambda(v_i)=\begin{cases}
		1 & i\in \{1,5,6\}, \\
		0 & otherwise. \\
	\end{cases}$. We call $s_\lambda$ a \textit{global-event-corresponding set} and use $S_{\Lambda}(\mathcal{E})$ to denote the set of all $s_{\lambda}$ on event set $\mathcal{E}$ and $S_{\Lambda}(v)=\{s_{\lambda}|v\in s_{\lambda}\}$.
\end{notation}

According to Proposition~\ref{prop3}, we can only choose from $V_3$, while selecting '$v$' in Definition~\ref{def7}. Otherwise, a logical contextual quantum state on Yu-Oh scenario may also be found, but this can not be witnessed simply by vectors in Yu-Oh set. We do not consider such cases and they are not helpful when constructing Hardy's paradox for Yu-Oh set. We use an algorithm, given in the appendix along with the table of all global events on Yu-Oh set, to list out $S_{\Lambda}(v)$ for all $v\in V_3$.

\begin{equation}
	\tag{*}
	\label{globalevents}
	\begin{aligned}
		S_{\Lambda}(v_A)= & \{ \{v_1,v_5,v_6,v_A\},\{v_2,v_6,v_7,v_A\},\{v_3,v_5,v_7,v_A\}\} \\
		S_{\Lambda}(v_B)= & \{ \{v_1,v_6,v_8,v_B\},\{v_2,v_4,v_6,v_B\},\{v_3,v_4,v_8,v_B\}\} \\
		S_{\Lambda}(v_C)= & \{ \{v_1,v_5,v_9,v_C\},\{v_2,v_4,v_9,v_C\},\{v_3,v_4,v_5,v_C\}\} \\
		S_{\Lambda}(v_D)= & \{ \{v_1,v_8,v_9,v_D\},\{v_2,v_7,v_9,v_D\},\{v_3,v_7,v_8,v_D\}\}
	\end{aligned}
\end{equation}

To search for logically contextual pure states, we need to let pure state $|\psi\rangle$, corresponding to quantum state $\rho=|\psi\rangle\langle\psi|$, satisfy that there exists vector $v_k\in V_3$ such that $\rho(v_k)=tr(|\psi\rangle\langle \psi|\frac{|v_k\rangle\langle v_k|}{\|v_k\|^2})=\frac{\|\langle v_k|\psi\rangle\|^2}{\|v_k\|^2}>0$, which implies $\|\langle v_k|\psi\rangle\|\neq0$. And for all $s_{\lambda}^i\in S_{\Lambda}(v_k)$, $i=\{1,2,3\}$ as shown in (\ref{globalevents}), there is $v_{k_i}\in s_{\lambda}^i$ such that $\rho(v_{k_i})=tr(|\psi\rangle\langle \psi|\frac{|v_{k_i}\rangle\langle v_{k_i}|}{\|v_{k_i}\|^2})=0$, which implies $\langle v_{k_i}|\psi\rangle=0$. Thus for $k\in\{A,B,C,D\}$ we only need to enumerate all possible multiset $\{v_{k_i}|v_{k_i}\in s_{\lambda}^i, s_{\lambda}^i\in S_{\Lambda}(v_k), i=\{1,2,3\}\}$, and solve linear equations $\begin{cases}
	\langle v_{k_1}|\psi\rangle=0 \\
	\langle v_{k_2}|\psi\rangle=0 \\
	\langle v_{k_3}|\psi\rangle=0 \\
\end{cases}$ to obtain some non-zero pure states. Then we can select the ones satisfying $\rho(v_k)=\|\langle v_k|\psi\rangle\|\neq0$, so as to exhaust pure states with logical contextuality on Yu-Oh set.

We use the following algorithm to realize the calculation above. It is worth mentioning that the calculations performed by our proposed algorithm are all conducted in the complex domain $\mathbb{C}$.

\begin{algorithm}[H]
	\caption{Algorithm to exhaust pure states with logical contextuality on Yu-Oh set}
	\label{alg1}
	\renewcommand{\algorithmicrequire}{\textbf{Input:}}
	\renewcommand{\algorithmicensure}{\textbf{Output:}}
	
	\begin{algorithmic}[1]
		\Require vector set: $V_3$, $\{S_{\Lambda}(v_k)|v_k\in V_3\}$ 
		\Ensure set of logically contextual pure states on Yu-Oh set: $\Psi$
		
		\State $\Psi=[]$
		
		\For{each $v_k\in V_3$}
		\For{each $(v_{k_1},v_{k_2},v_{k_3})|v_{k_i}\in s_{\lambda}^i$, $s_{\lambda}^i\in S_{\Lambda}(v_k)$, $i\in\{1,2,3\}$}
		\State solve linear equations: $\begin{cases}
			\langle v_{k_1}|\psi\rangle=0 \\
			\langle v_{k_2}|\psi\rangle=0 \\
			\langle v_{k_3}|\psi\rangle=0 \\
			\end{cases}$
		\If{$\psi\neq 0$ and $\langle v_k|\psi\rangle\neq 0$}
		\State add $|\psi\rangle$ into $\Psi$
		\EndIf
		\EndFor
		\EndFor
		
		\Return $\Psi$
		
	\end{algorithmic}
\end{algorithm}

We finally achieve the conclusion below.

\begin{theorem}
	\label{the2}
	The set of all pure states, omitted normalization, being logically contextual on Yu-Oh set is as follows.
	\begin{equation*}
		\{(1,1,1)^T,(1,-1,1)^T,(1,1,-1)^T,(-1,1,1)^T\}
	\end{equation*}
\end{theorem}

\subsection{From pure states to general quantum states}\label{subsec3.2} 

\begin{notation}
	A mixed state is denoted as $\rho=\sum\limits_{i=1}^np_i|\psi_i\rangle\langle\psi_i|$, where $p_i>0$ and $\sum\limits_ip_i=1$. Pure states $|\psi_i\rangle$ corresponds to density operator $\rho_i=|\psi_i\rangle\langle\psi_i|$.
\end{notation}

In this section, we investigate the relation between the possibilistic property of a certain event under a mixed state $\rho=\sum\limits_{i=1}^np_i|\psi_i\rangle\langle\psi_i|$ and it under all pure states $|\psi_i\rangle$,  so as to use the structure of pure states to describe mixed states being logically contextual on Yu-Oh set. 

\begin{proposition}
	\label{prop5}
	Given mixed state $\rho=\sum\limits_{i=1}^np_i|\psi_i\rangle\langle\psi_i|$ and quantum event $v\neq0$, $\overline{\rho}(v)=0$ if and only if $\overline{\rho_i}(v)=0$, for all $i \in \{1,2,...,n\}$, where $\rho_i=|\psi_i\rangle\langle\psi_i|$.
\end{proposition}

\begin{proof}
	\begin{align*}
		\|v\|^2\cdot\rho(v)=&tr(\rho|v\rangle\langle v|)=\langle v|\rho|v\rangle=\langle v|(\sum\limits_{i=1}^np_i|\psi_i\rangle\langle\psi_i|)|v\rangle=\sum\limits_{i=1}^np_i\langle v|\psi_i\rangle\langle\psi_i|v\rangle \\
		=&\sum\limits_{i=1}^np_i\|\langle v|\psi_i\rangle\|^2 = \|v\|^2\cdot\sum\limits_{i=1}^np_i\rho_i(v).
	\end{align*}
	
	Thus $\rho(v)=\sum\limits_{i=1}^np_i\rho_i(v)$. Due to the non-negativity of quantum probability, we can deduce that: $\rho(v)=0$ if and only if $\rho_i(v)=0$ for all $i \in \{1,2,...,n\}$. Thus $\overline{\rho}(v)=0$ if and only if $\overline{\rho_i}(v)=0$ for all $i \in \{1,2,...,n\}$.
\end{proof}

Consider the negation of Proposition~\ref{prop5}, and we can get:

\begin{corollary}
	\label{cor}
	Given mixed state $\rho=\sum\limits_{i=1}^np_i|\psi_i\rangle\langle\psi_i|$ and quantum event $v$. $\overline{\rho}(v)=1$ if and only if there exists $i\in\{1,2,...,n\}$ such that $\overline{\rho_i}(v)=1$.
\end{corollary}

We apply Proposition~\ref{prop5} and Corollary~\ref{cor} on Proposition~\ref{prop2}, so as to get Proposition~\ref{prop6}.

\begin{proposition}
	\label{prop6}
	Mixed state $\rho=\sum\limits_{i=1}^np_i|\psi_i\rangle\langle\psi_i|$ is logically contextual on vector set $\mathcal{E}$ is equivalent to the following conditions.
	\begin{itemize}
		\item There exists $v \in \mathcal{E}$, such that there exists $i \in \{1,2,...,n\}: \overline{\rho_i}(v)=1$.
		\item For all $\lambda \in \Lambda(\mathcal{E})$, $\lambda(v)=1$, there is $v' \in \mathcal{E}$, $\lambda(v')=1$, such that $\overline{\rho_i}(v')=0$.
	\end{itemize}
\end{proposition}

To characterize the structure of logically contextual mixed states on Yu-Oh set, we first give notations as follows.
\begin{notation}
	\label{not5}
	Given vector $v_k$, select $v_{k_i}\in s_{\lambda}^i$, $s_{\lambda}^i\in S_{\Lambda}(v_k)$. We solve linear equations $\begin{cases}
	\langle v_{k_1}|\psi\rangle=0 \\
	\langle v_{k_2}|\psi\rangle=0 \\
	\langle v_{k_3}|\psi\rangle=0 \\
\end{cases}$ to obtain its solution space $\Psi_k$, the vector in which we categorize into 2 classes, namely, $M=\{m\big|\overline{\rho_m}(v_k)=1, \rho_m=|\psi_m\rangle\langle\psi_m|\}$, and $N=\{n\big|\overline{\rho_n}(v_k)=0, \rho_n=|\psi_n\rangle\langle\psi_n|\}$. $M \cap N=\emptyset$ and $\{\psi_m\}_{m\in M} \cup \{\psi_n\}_{n\in N}=\Psi_k$.
\end{notation}

Proposition~\ref{prop6} gives the structure of logically contextual mixed states on Yu-Oh set. Using the notations given in Notation~\ref{not5}, it can be described as Proposition~\ref{prop7}:
\begin{proposition}
	\label{prop7}
	The mixed states being logically contextual on Yu-Oh set must have structure as:
	\begin{equation*}
		\rho=\sum\limits_{m\in M}p_m|\psi_m\rangle\langle\psi_m|+\sum\limits_{n\in N}p_n|\psi_n\rangle\langle\psi_n|,
	\end{equation*}
	where $k\in\{A,B,C,D\}$, $p_m,p_n\in[0,1]$,  $\sum\limits_mp_m+\sum\limits_np_n=1$. And there exists $p_{m_0} \neq 0$ and $m_o\in M$.
\end{proposition}

\begin{proof}
	Suppose $\rho$ has the aforementioned structure. The probability of vector $v_k$ under $\rho$ is greater than $0$, thus:
	\begin{align*}
		\rho(v_k)&=tr(\rho\frac{|v_k\rangle\langle v_k|}{\|v_k\|^2})=\frac{1}{\|v_k\|^2}\cdot\langle v_k|\rho|v_k\rangle \\
		&=\frac{1}{\|v_k\|^2}\cdot(\sum\limits_{m\in M}p_m\langle v_k|\psi_m\rangle\langle\psi_m|v_k\rangle+\sum\limits_{n\in N}p_n\langle v_k|\psi_n\rangle\langle\psi_n|v_k\rangle) \\
		&=\sum\limits_{m\in M}p_m\rho_m(v_k)+\sum\limits_{n\in N}p_n\rho_n(v_k) \ge p_{m_0}\rho_{m_0}(v_k)>0
	\end{align*}

	The probability of vector $v_{k_i}$ under $\rho$, $i\in\{1,2,3\}$:
	\begin{equation*}
		\rho(v_{k_i})=\sum\limits_{m\in M}p_m\rho_m(v_{k_i})+\sum\limits_{n\in N}p_n\rho_n(v_{k_i})=0
	\end{equation*}
	Then it is verified that $\rho$ is logically contextual on Yu-Oh set.
	
	On the other hand, if $\rho=\sum\limits_{j\in J}p_j|\psi_j\rangle\langle\psi_j|$ is logically contextual on Yu-Oh set, according to Proposition~\ref{prop2}:
	
	(1)there exists vector $v_k \in \{v_A,v_B,v_C,v_D\}$ and $j_0\in J$, such that $p_{j_0}>0$ and $\overline{\rho_{j_0}}(v_k)=1$.
	
	(2)For all global events $s_{\lambda}^i\in S_{\Lambda}(v_k), S_{\Lambda}(v_k)=\{s_{\lambda}^i\}_{i=1}^3$, there is $v_{k_i}\in s_{\lambda}^i$, such that $\overline{\rho_j}(v_{k_i})=\lceil \frac{\langle v_{k_i}|\psi_j\rangle\langle\psi_j|v_{k_i}\rangle}{\|v_{k_i}\|^2} \rceil=0$ for all $j\in J$. Thus $|\psi_j\rangle$ lies in the solution space of linear equations:
	\begin{equation*}
		\begin{cases}
			\langle v_{k_1}|\psi\rangle=0 \\
			\langle v_{k_2}|\psi\rangle=0 \\
			\langle v_{k_3}|\psi\rangle=0 \\
		\end{cases}.
	\end{equation*} 
	
	According to (2), $\rho$ have the structure as:
	\begin{equation*}
		\rho=\sum\limits_{m\in M}p_m|\psi_m\rangle\langle\psi_m|+\sum\limits_{n\in N}p_n|\psi_n\rangle\langle\psi_n|,
	\end{equation*}
	where $p_m,p_n\in[0,1]$,  $\sum\limits_mp_m+\sum\limits_np_n=1$.
	
	According to (1), there exists $j_0\in J$, such that $p_{j_0}>0$ and $\overline{\rho_{j_0}}(v_k)=1$, $v_k\in\{v_A,v_B,v_C,v_D\}$. Let $|\psi_{m_0}\rangle=|\psi_{j_0}\rangle$, then there is $p_{m_0}=p_{j_0}>0$ and $\psi_{m_0}\in \Psi_M$.
\end{proof}

But by analyzing the computing process of Algorithm~\ref{alg1}, it is discovered that the dimension of  solution space for linear equations $\begin{cases}
	\langle v_{k_1}|\psi\rangle=0 \\
	\langle v_{k_2}|\psi\rangle=0 \\
	\langle v_{k_3}|\psi\rangle=0 \\
\end{cases}$ is at most 1. Thus lead to the proof of Proposition~\ref{prop8}.

\begin{proposition}
	\label{prop8}
	None of mixed state is logically contextual on the Yu-Oh set.
\end{proposition}

\begin{proof}
	As shown in (\ref{globalevents}), for all vector $v_k\in V_3$ and $S_{\Lambda}(v_k)$, there is not any vector $v_k' \in \mathcal{E}_{Yu-Oh}$, $v_k' \neq v_k$, such that for all $s_{\lambda}^i\in S_{\Lambda}(v_k)$, $v_k'\in s_{\lambda}^i$. Thus multiset $\{v_{k_i}\}_{i=1}^3$ can not have three repeated elements. So linear equations $\begin{cases}
	\langle v_{k_1}|\psi\rangle=0 \\
	\langle v_{k_2}|\psi\rangle=0 \\
	\langle v_{k_3}|\psi\rangle=0 \\
\end{cases}$ contains at least two equations that is linearly independent. The rank of coefficient matrix for the equations is at least 2. Thus the dimension of its solution space is at most 1, which indicates that quantum state $\rho$ can only be a pure state. 
\end{proof}

Together with Theorem~\ref{the2} and Proposition~\ref{prop8}, we get the theorem that gives all quantum states with logical contextuality on Yu-Oh set as follows.

\begin{theorem}
	All logically contextual quantum states on Yu-Oh set are pure states, which are, omitted normalization:
	\begin{equation*}
		\{(1,1,1)^T,(1,-1,1)^T,(1,1,-1)^T,(-1,1,1)^T\}
	\end{equation*}
\end{theorem}

\section{Hardy's paradoxes for Yu-Oh set}\label{sec4}

The concept of Hardy's paradox originates from Hardy's argument. It is now generally referred to as a set of possibilistic conditions leading to contextuality in an inequality-free approach. 

Hardy's paradox, in the bipartite Bell-scenario, is based on 4 possibilistic conditions as follows:

\begin{align*}
	\rho(1,1|a,b)>0 \quad & \rho(1,1|a,b')=0 \\
	\rho(1,1|a',b)=0 \quad & \rho(0,0|a',b')=0
\end{align*}

According to these conditions, notice that $a=1 \implies b'=0$ and $b=1 \implies a'=0$, from which we have $\rho(1,1|a,b)>0 \implies \rho(0,0|a',b')>0$, leading to contradiction. We can also present the conclusion in Section~\ref{sec3} as similar paradoxes.

Lets take $|\psi\rangle=(1,-1,1)^T$ being logically contextual on Yu-Oh set as an example. It is easily verified that vectors $v_5$ and $v_7$ satisfy possibilistic conditions $\rho(v_5)=\rho(v_7)=0$ and $\rho(v_A)>0$, where $\rho=|\psi\rangle\langle\psi|$. It can be proved that these conditions lead to contradiction, thus generated as Proposition~\ref{prop10} as follows.

\begin{proposition}
	\label{prop10}
	Quantum state $\rho$ being logically contextual on $\mathcal{E}_{Yu-Oh}$ implies Hardy's paradoxes as follows.
	\begin{align*}
		\rho(v_k)>0 \quad \rho(v_{k_i})=0
	\end{align*}
	, where $v_k\in V_3$, $v_{k_i} \in s_{\lambda}^i$ and $s_{\lambda}^i\in S_{\Lambda}(v_k), i\in \{1,2,3\}$.
\end{proposition}

\begin{proof}
	Suppose $\rho$ is logically contextual on $\mathcal{E}_{Yu-Oh}$. Then there is $v_k\in\{v_A,v_B,v_C,v_D\}$ satisfying $\rho(v_k)>0$ and for all $s_{\lambda}^i \in S_{\Lambda}(v_k)$, $i \in \{1,2,3\}$, there is $v_{k_i} \in s_{\lambda}^i$, $v_{k_i} \neq v_k$ such that $\rho(v_{k_i})=0$. Thus a set of possibilistic conditions can be generated as follows:
	\begin{align*}
		\rho(v_k)>0 \quad \rho(v_{k_i})=0, i\in \{1,2,3\}.
	\end{align*}
	
	Now we prove the above conditions lead to contradiction.
	
	Under classical probability theory, there exists a global possibilistic distribution $p_{\Lambda}$ for global-event-corresponding sets $S_{\Lambda}(\mathcal{E})$, and it is generated by $\{\overline{\rho}(v)|v\in\mathcal{E}_{Yu-Oh}\}$. Then, for vector $v_k$ in context $C_k\in\mathcal{M}(\mathcal{E}_{Yu-Oh})$:
	\begin{equation*}
		\overline{\rho}(v_k)=p_{\Lambda}|_{C_k}(v_k)=\bigvee\limits_{s_{\lambda_m} \in S_{\Lambda}(v_k)}p_{\Lambda}(s_{\lambda_m}).
	\end{equation*}
	
	Yet we also have: for all $i \in \{1,2,3\}:$
	\begin{equation*}
		\overline{\rho}(v_{k_i})=\bigvee\limits_{s_{\lambda_n} \in S_{\Lambda}(v_{k_i})}p_{\Lambda}(s_{\lambda_n})=\lceil \rho(v_{k_i})\rceil=0,
	\end{equation*}
	which implies that for all $s_{\lambda_n} \in S_{\Lambda}(v_{k_i})$, $p_{\Lambda}(s_{\lambda_n})=0$.
	Due to the fact that for all $m$, $s_{\lambda_m}$ is also in $S_{\Lambda}(v_{k_i})$, so $p_{\Lambda}(s_{\lambda_m})=0$. Thus $\overline{\rho}(v_k)=\bigvee\limits_{s_{\lambda_m} \in S_{\Lambda}(v_k)}p_{\Lambda}(s_{\lambda_m})=0$, which leads to contradiction. In a nutshell, conditions $\rho(v_k)>0 \quad \rho(v_{k_i})=0, i\in \{1,2,3\}$ imply a Hardy's paradox.
\end{proof}

Specifically, for each of the four quantum states listed in Section~\ref{sec3}, which are logically contextual on Yu-Oh set, we can construct the following 12 Hardy's paradoxes under each state after calculation.

\setcounter{equation}{0}
\begin{align}
	\label{hp}
		(1,1,1)^T:\quad &\rho(v_A)>0, \rho(v_5)=\rho(v_6)=0 \\
		(1,1,1)^T:\quad &\rho(v_B)>0, \rho(v_4)=\rho(v_6)=0 \\
		(1,1,1)^T:\quad &\rho(v_C)>0, \rho(v_4)=\rho(v_5)=0 \\
		(-1,1,1)^T:\quad &\rho(v_B)>0, \rho(v_4)=\rho(v_8)=0 \\
		(-1,1,1)^T:\quad &\rho(v_C)>0, \rho(v_4)=\rho(v_9)=0 \\
		(-1,1,1)^T:\quad &\rho(v_D)>0, \rho(v_8)=\rho(v_9)=0 \\
		(1,-1,1)^T:\quad &\rho(v_A)>0, \rho(v_5)=\rho(v_7)=0 \\
		(1,-1,1)^T:\quad &\rho(v_C)>0, \rho(v_5)=\rho(v_9)=0 \\
		(1,-1,1)^T:\quad &\rho(v_D)>0, \rho(v_7)=\rho(v_9)=0 \\
		(1,1,-1)^T:\quad &\rho(v_A)>0, \rho(v_6)=\rho(v_7)=0 \\
		(1,1,-1)^T:\quad &\rho(v_B)>0, \rho(v_6)=\rho(v_8)=0 \\
		(1,1,-1)^T:\quad &\rho(v_D)>0, \rho(v_7)=\rho(v_8)=0
\end{align}

As listed above, we finally construct 12 Hardy's paradoxes for Yu-Oh set with the help of logically contextual quantum states. Thus supporting the Proposition~\ref{prop11}.

\begin{proposition}
	\label{prop11}
	Hardy's paradox is necessary condition of logical contextuality in Yu-Oh set. 
\end{proposition}

Success probability $SP$ is a vital evaluation parameter for Hardy's paradox. It characterizes the probability of witnessing the corresponding possibilistic conditions for a certain Hardy's paradox in an experiment. Its generalized notion, degree of success $DS$, was proposed in \cite{Rai2021} to cater for more generalized yet inequality-dependent Hardy-type paradoxes proposed by Cabello \cite{Cabello2002}, Liang and Li \cite{Liang2003,Liang2005}. Our constructions of Hardy's paradoxes for Yu-Oh set give inequality-free proofs for contextuality, so we can give $SP$s of Hardy's paradoxes $(1-12)$. Denote $SP$ of Hardy's paradox $(i)$ as $SP_i$, $i\in\{1,2,...,12\}$, then we have the following results.

\begin{equation*}
	SP_i=
	\begin{cases}
		\rho(v_A)=11.1\%,&i=1,7,10 \\
		\rho(v_B)=11.1\%,&i=2,4,11 \\
		\rho(v_C)=11.1\%,&i=3,5,8 \\
		\rho(v_D)=11.1\%,&i=6,9,12 \\
	\end{cases}
\end{equation*}

\section{Supporting observables}\label{sec5}

Based on the given Hardy's paradoxes in Section~\ref{sec4}, we can now propose their corresponding observables to witness each of them in experiments. As shown in Section~\ref{sec4}, it is noticed that in each Hardy's paradoxes, the possiblistic conditions are consisted of one vector with probability greater than 0 and two equal to 0. Such paradoxes can be witnessed if we can find observables indicating such conditions.

\begin{proposition}
	\label{prop12}
	Given the scenario $\mathcal{E}_{Yu-Oh}$ and quantum state $\rho=|\psi\rangle\langle\psi|$, $|\psi\rangle$ being real, only 1 observable is needed to witness possiblistic conditions as below.
	\begin{equation*}
		\rho(v_{k_1})>0, \quad \rho(v_{k_2})=\rho(v_{k_3})=0,
	\end{equation*}
	where $v_{k_1},v_{k_2},v_{k_3}\in \mathcal{E}_{Yu-Oh}$ are linearly independent vectors and are \textit{normalized}\footnote{It is not hard to verify the vectors in each Hardy's paradoxes constructed in Section~\ref{sec4} are linearly independent.}.
	
	The observable has unique spectral decomposition as $\hat{A}=\sum\limits_{i=1}^3a_i\frac{|v_{k_i}'\rangle\langle v_{k_i}'|}{\|v_{k_i}'\|^2}$, where $a_1$, $a_2$, $a_3$ are real numbers and $v_{k_1}'$, $v_{k_2}'$, $v_{k_3}'$ are $v_{k_1}$, $v_{k_2}$, $v_{k_3}$ after Gram-Schmidt orthogonalizition process.
\end{proposition}

\begin{proof}
	Given possiblistic conditions as $\rho(v_{k_1})>0, \quad \rho(v_{k_2})=\rho(v_{k_3})=0$, we have($\{v_{k_i}\}_{i=1}^3$ are normalized):
	\begin{align*}
		&\rho(v_{k_1})=\langle\psi|v_{k_1}\rangle\langle v_{k_1}|\psi\rangle=\|\langle\psi|v_{k_1}\rangle\|^2>0 \\
		&\rho(v_{k_2})=\|\langle\psi|v_{k_2}\rangle\|^2=0 \\
		&\rho(v_{k_3})=\|\langle\psi|v_{k_3}\rangle\|^2=0
	\end{align*}
	Thus we have $\langle\psi|v_{k_2}\rangle=\langle\psi|v_{k_3}\rangle=0$.
	
	We firstly perform Gram-Schmidt orthogonalizition process on $v_{k_1}$, $v_{k_2}$ and $v_{k_3}$ as follows.
	\begin{align*}
		&v_{k_1}'=v_{k_3} \\
		&v_{k_2}'=v_{k_2}-\frac{(v_{k_2},v_{k_1}')}{||v_{k_1}'||^2}v_{k_1}' \triangleq v_{k_2}+l_1v_{k_1}' \\
		&v_{k_3}'=v_{k_1}-\frac{(v_{k_2},v_{k_1}')}{||v_{k_1}'||^2}v_{k_1}'-\frac{(v_{k_3},v_{k_2}')}{||v_{k_2}'||^2}v_{k_2}' \triangleq v_{k_1}+l_1v_{k_1}'+l_2v_{k_2}'
	\end{align*}
	$l_1,l_2$ are real because $v_{k_2},v_{k_3}$ are real.

	We can construct a set of mutually orthogonal projectors $P_i=\frac{|v_i'\rangle\langle v_i'|}{\|v_i'\|^2},i=1,2,3$. Let observable $\hat{A}$ has spectral decomposition as  $\hat{A}=\sum\limits_{i=1}^3a_iP_i$, where $a_1,a_2,a_3$ are real numbers, and we will have:
	\begin{align*}
		&\rho(\hat{A}=a_1)=\rho(P_1)=\rho(v_{k_3})=0 \\
		&\begin{aligned}
			\rho(\hat{A}=a_2)&=\rho(P_2)=\frac{1}{\|v_{k_2}\|^2}\langle\psi|(|v_{k_2}\rangle+l_1|v_{k_3}\rangle)(\langle v_{k_2}|+l_1\langle v_{k_3}|)|\psi\rangle \\
			&=\frac{1}{\|v_{k_2}\|^2}(\langle\psi|v_{k_2}\rangle\langle v_{k_2}|\psi\rangle+l_1\langle\psi|v_{k_2}\rangle\langle v_{k_3}|\psi\rangle+l_1\langle\psi|v_{k_3}\rangle\langle v_{k_2}|\psi\rangle+\langle\psi|v_{k_3}\rangle\langle v_{k_3}|\psi\rangle)=0
		\end{aligned} \\
		&\rho(\hat{A}=a_3)=1-\rho(\hat{A}=a_1)-\rho(\hat{A}=a_2)=1
	\end{align*}
	$\hat{A}=a_1,\hat{A}=a_2,\hat{A}=a_3$ are collectively exhaustive events in $3$-dimensional Hilbert space.

	On the other hand, assume $\rho(\hat{A}=a_1)=\rho(\hat{A}=a_2)=0, \rho(\hat{A}=a_3)=1$, we have:
	\begin{align*}
		&\rho(v_{k_3})=\rho(\hat{A}=a_1)=(\langle\psi|v_{k_3}\rangle)^2=0 \\
		&\begin{aligned}
			\rho(\hat{A}=a_2)
			&=\frac{1}{\|v_{k_2}\|^2}(\langle\psi|v_{k_2}\rangle\langle v_{k_2}|\psi\rangle+l_1\langle\psi|v_{k_2}\rangle\langle v_{k_3}|\psi\rangle+l_1\langle\psi|v_{k_3}\rangle\langle v_{k_2}|\psi\rangle+\langle\psi|v_{k_3}\rangle\langle v_{k_3}|\psi\rangle) \\
			&=\frac{1}{\|v_{k_2}\|^2}[(\langle\psi|v_{k_2}\rangle+l_1\langle\psi|v_{k_3}\rangle)^2+(1-l_1^2)(\langle\psi|v_{k_3}\rangle)^2] \\
			&=\frac{1}{\|v_{k_2}\|^2}(\langle\psi|v_{k_2}\rangle+l_1\langle\psi|v_{k_3}\rangle)^2=0
		\end{aligned}
	\end{align*}
	which implies that $\langle\psi|v_{k_3}\rangle=\langle\psi|v_{k_2}\rangle=0$. Thus $\rho(v_{k_3})=\rho(v_{k_2})=0$.
	
	Since $v_{k_1},v_{k_2},v_{k_3}$ are linearly independent in 3-dimensional Hilbert space, there is no $|\psi\rangle\neq 0$ such that $\langle\psi|v_i\rangle=0$ for all $i\in \{1,2,3\}$, which implies that $\rho(v_{k_1})=(\langle\psi|v_{k_1}\rangle)^2>0$.
	
	 As the calculation shows, the inevitability of event $\hat{A}=a_3$ witnesses possiblistic conditions as $\rho(v_{k_1})>0, \quad \rho(v_{k_2})=\rho(v_{k_3})=0$.
\end{proof}

With Proposition~\ref{prop12} we can construct observables to witness Hardy's paradoxes $(1-12)$. Take Hardy's paradox $(1)$ for an example.

Given scenario $\mathcal{E}_{Yu-Oh}$ and quantum state $|\psi\rangle=\frac{1}{\sqrt{3}}(1,1,1)^T$, possiblistic conditions $\rho(v_A)>0,\rho(v_5)=\rho(v_6)=0$ construct a Hardy's paradox. According to Proposition~\ref{prop12}, we can construct projectors as follows.

\[
P_1=\frac{1}{2}\begin{pmatrix}
	1 & 0 & -1 \\
	0 & 0 & 0 \\
	-1 & 0 & 1
\end{pmatrix},\quad
P_2=\frac{1}{6}\begin{pmatrix}
	1 & -2 & 1 \\
	-2 & 4 & -2 \\
	1 & -2 & 1
\end{pmatrix},\quad
P_3=\frac{1}{3}\begin{pmatrix}
	1 & 1 & 1 \\
	1 & 1 & 1 \\
	1 & 1 & 1
\end{pmatrix}
\]

Let observable $\hat{A}=a_1P_1+a_2P_2+a_3P_3$, where $a_1,a_2,a_3$ are real numbers. Then the inevitability of event $\hat{A}=a_3$ witnesses Hardy-type paradox (1). Here we provide projectors $P_1,P_2,P_3$ to experimentally witness Hardy's paradoxes $(1-12)$ in Table~\ref{table1}.

\begin{longtable}{c|c|ccc}
	\label{table1}
    Paradox & State $|\psi\rangle$ & $P_1$ & $P_2$ & $P_3$ \\ [2pt]
    \hline
$(1)$ & $\frac{1}{\sqrt{3}}(1,1,1)^T$ & 
$\frac{1}{2}\begin{pmatrix}
	1 & 0 & -1 \\
	0 & 0 & 0 \\
	-1 & 0 & 1
\end{pmatrix}$ & 
$\frac{1}{6}\begin{pmatrix}
	1 & -2 & 1 \\
	-2 & 4 & -2 \\
	1 & -2 & 1
\end{pmatrix}$ & 
$\frac{1}{3}\begin{pmatrix}
	1 & 1 & 1 \\
	1 & 1 & 1 \\
	1 & 1 & 1
\end{pmatrix}$ \\ [1pt]
$(2)$ & $\frac{1}{\sqrt{3}}(1,1,1)^T$ & 
$\frac{1}{2}\begin{pmatrix}
	0 & 0 & 0 \\
	0 & 1 & -1 \\
	0 & -1 & 1
\end{pmatrix}$ & 
$\frac{1}{6}\begin{pmatrix}
	4 & -2 & -2 \\
	-2 & 1 & 1 \\
	-2 & 1 & 1
\end{pmatrix}$ & 
$\frac{1}{3}\begin{pmatrix}
	1 & 1 & 1 \\
	1 & 1 & 1 \\
	1 & 1 & 1
\end{pmatrix}$ \\ [1pt]
$(3)$ & $\frac{1}{\sqrt{3}}(1,1,1)^T$ & 
$\frac{1}{2}\begin{pmatrix}
	0 & 0 & 0 \\
	0 & 1 & -1 \\
	0 & -1 & 1
\end{pmatrix}$ & 
$\frac{1}{6}\begin{pmatrix}
	4 & -2 & -2 \\
	-2 & 1 & 1 \\
	-2 & 1 & 1
\end{pmatrix}$ & 
$\frac{1}{3}\begin{pmatrix}
	1 & 1 & 1 \\
	1 & 1 & 1 \\
	1 & 1 & 1
\end{pmatrix}$ \\ [1pt]
$(4)$ & $\frac{1}{\sqrt{3}}(-1,1,1)^T$ & 
$\frac{1}{2}\begin{pmatrix}
	0 & 0 & 0 \\
	0 & 1 & -1\\
	0 & -1 & 1
\end{pmatrix}$ & 
$\frac{1}{6}\begin{pmatrix}
	4 & 2 & 2 \\
	2 & 1 & 1 \\
	2 & 1 & 1
\end{pmatrix}$ & 
$\frac{1}{2}\begin{pmatrix}
	0 & 0 & 0 \\
	0 & 1 & 1 \\
	0 & 1 & 1
\end{pmatrix}$ \\ [1pt]
$(5)$ & $\frac{1}{\sqrt{3}}(-1,1,1)^T$ & 
$\frac{1}{2}\begin{pmatrix}
	0 & 0 & 0 \\
	0 & 1 & -1 \\
	0 & -1 & 1
\end{pmatrix}$ & 
$\frac{1}{6}\begin{pmatrix}
	4 & 2 & 2 \\
	2 & 1 & 1 \\
	2 & 1 & 1
\end{pmatrix}$ & 
$\frac{1}{2}\begin{pmatrix}
	0 & 0 & 0 \\
	0 & 1 & 1 \\
	0 & 1 & 1
\end{pmatrix}$ \\ [1pt]
$(6)$ & $\frac{1}{\sqrt{3}}(-1,1,1)^T$ & 
$\frac{1}{2}\begin{pmatrix}
	1 & 0 & 1 \\
	0 & 0 & 0 \\
	1 & 0 & 1
\end{pmatrix}$ & 
$\frac{1}{6}\begin{pmatrix}
	1 & 2 & -1 \\
	2 & 4 & -2 \\
	-1 & -2 & 1
\end{pmatrix}$ & 
$\frac{1}{3}\begin{pmatrix}
	-1 & -1 & -1 \\
	-1 & 1 & 1 \\
	-1 & 1 & 1
\end{pmatrix}$ \\ [1pt]
$(7)$ & $\frac{1}{\sqrt{3}}(1,-1,1)^T$ & 
$\frac{1}{2}\begin{pmatrix}
	1 & 0 & -1 \\
	0 & 0 & 0 \\
	-1 & 0 & 1
\end{pmatrix}$ & 
$\frac{1}{6}\begin{pmatrix}
	1 & -2 & 1 \\
	-2 & 4 & -2 \\
	1 & -2 & 1
\end{pmatrix}$ & 
$\frac{1}{3}\begin{pmatrix}
	1 & 1 & 1 \\
	1 & 1 & 1 \\
	1 & 1 & 1
\end{pmatrix}$ \\ [1pt]
$(8)$ & $\frac{1}{\sqrt{3}}(1,-1,1)^T$ & 
$\frac{1}{2}\begin{pmatrix}
	1 & 0 & -1 \\
	0 & 0 & 0 \\
	-1 & 0 & 1
\end{pmatrix}$ & 
$\frac{1}{6}\begin{pmatrix}
	1 & 2 & 1 \\
	2 & 4 & 2 \\
	1 & 2 & 1
\end{pmatrix}$ & 
$\frac{1}{3}\begin{pmatrix}
	1 & -1 & 1 \\
	-1 & 1 & -1 \\
	1 & -1 & 1
\end{pmatrix}$ \\ [1pt]
$(9)$ & $\frac{1}{\sqrt{3}}(1,-1,1)^T$ & 
$\frac{1}{2}\begin{pmatrix}
	0 & 0 & 0 \\
	0 & 1 & 1 \\
	0 & 1 & 1
\end{pmatrix}$ & 
$\frac{1}{6}\begin{pmatrix}
	4 & 2 & -2 \\
	2 & 1 & -1 \\
	-2 & -1 & 1
\end{pmatrix}$ & 
$\frac{1}{3}\begin{pmatrix}
	1 & -1 & 1 \\
	-1 & 1 & -1 \\
	1 & -1 & 1
\end{pmatrix}$ \\ [1pt]
$(10)$ & $\frac{1}{\sqrt{3}}(1,1,-1)^T$ & 
$\frac{1}{2}\begin{pmatrix}
	1 & -1 & 0 \\
	-1 & 1& 0 \\
	0 & 0 & 0
\end{pmatrix}$ & 
$\frac{1}{6}\begin{pmatrix}
	1 & 1 & 2 \\
	1 & 1 & 2 \\
	2 & 2 & 4
\end{pmatrix}$ & 
$\frac{1}{3}\begin{pmatrix}
	1 & 1 & -1 \\
	1 & 1 & -1 \\
	-1 & -1 & 1
\end{pmatrix}$ \\ [1pt]
$(11)$ & $\frac{1}{\sqrt{3}}(1,1,-1)^T$ & 
$\frac{1}{2}\begin{pmatrix}
	1 & -1 & 0 \\
	-1 & 1 & 0 \\
	0 & 0 & 0
\end{pmatrix}$ & 
$\frac{1}{6}\begin{pmatrix}
	1 & 1 & 2 \\
	1 & 1 & 2 \\
	2 & 2 & 4
\end{pmatrix}$ & 
$\frac{1}{3}\begin{pmatrix}
	1 & 1 & -1 \\
	1 & 1 & -1 \\
	-1 & -1 & 1
\end{pmatrix}$ \\ [1pt]
$(12)$ & $\frac{1}{\sqrt{3}}(1,1,-1)^T$ & 
$\frac{1}{2}\begin{pmatrix}
	0 & 0 & 0 \\
	0 & 1 & 1 \\
	0 & 1 & 1
\end{pmatrix}$ & 
$\frac{1}{6}\begin{pmatrix}
	4 & -2 & 2 \\
	-2 & 1 & -1 \\
	2 & -1 & 1
\end{pmatrix}$ & 
$\frac{1}{3}\begin{pmatrix}
	1 & 1 & -1 \\
	1 & 1 & -1 \\
	-1 & -1 & 1
\end{pmatrix}$ \\
\caption{Table of projectors to construct observables that witnesses Hardy-type paradoxes $(1-12)$.}
\end{longtable} 

\section{Conclusion}\label{sec6}

The primary aim of this research is to construct Hardy's paradox for Yu–Oh set. We first identified all logically contextual states in Yu-Oh quantum system. In order to overcome the difficulty mentioned in \cite{ksgeneral2022} that sheaf-theoretic framework fail to recognize Yu-Oh's argument, we propose that the 13 vectors constituting the Yu–Oh set can be view as quantum events, an idea rooted in quantum logic \cite{birkhoff1936}. Furthermore, we demonstrate that each KS-assignment of projectors in Yu-Oh set induces a global event, so that we can search for logically contextual quantum states utilizing Proposition~\ref{prop2}.

By means of Proposition~\ref{prop3} and Algorithm~\ref{alg1}, we systematically identify all logically contextual pure states, while theoretically proving, via Proposition~\ref{prop5} to \ref{prop8}, that no mixed state can exhibit logical contextuality. Based on the identified logically contextual quantum states, we construct 12 Hardy's paradoxes in the scenario, each one of them with success probability $SP$ of 11.1\%. Finally, for each Hardy's paradox, we construct a single measurement to experimentally witness its possibilistic conditions.

\begin{appendices}
\floatname{algorithm}{Algorithm A}
\setcounter{algorithm}{0}

\section{Algorithms and related tables}\label{secA}

\begin{longtable}{l|l}
	\label{table:appendix1}
	Context $C\in\mathcal{C}(\mathcal{E})$ of 2 vectors & $C^{\perp}$ \\
	\hline
	$v_4=(0,1,-1)^T$, $v_A=(-1,1,1)$ & $(2,1,1)^T$ \\
	$v_8=(1,0,1)^T$, $v_A=(-1,1,1)$ & $(-1,-2,1)^T$ \\
	$v_9=(1,1,0)^T$, $v_A=(-1,1,1)$ & $(1,-1,2)^T$ \\
	$v_5=(1,0,-1)^T$, $v_B=(1,-1,1)$ & $(-1,-2,-1)^T$ \\
	$v_7=(0,1,1)^T$, $v_B=(1,-1,1)$ & $(2,1,-1)^T$ \\
	$v_9=(1,1,0)^T$, $v_B=(1,-1,1)$ & $(1,-1,-2)^T$ \\
	$v_6=(1,-1,0)^T$, $v_C=(1,1,-1)$ & $(1,1,2)^T$ \\
	$v_7=(0,1,1)^T$, $v_C=(1,1,-1)$ & $(-2,1,-1)^T$ \\
	$v_8=(1,0,1)^T$, $v_C=(1,1,-1)$ & $(-1,2,1)^T$ \\
	$v_4=(0,1,-1)^T$, $v_D=(1,1,1)$ & $(2,-1,-1)^T$ \\
	$v_5=(1,0,-1)^T$, $v_D=(1,1,1)$ & $(1,-2,1)^T$ \\
	$v_6=(1,-1,0)^T$, $v_D=(1,1,1)$ & $(-1,-1,2)^T$ \\
	\caption{The orthogonal complete for context $C\in\mathcal{C}(\mathcal{E})$ of 2 vectors.}
\end{longtable}

\begin{algorithm}[H]
	\label{alg:appendix1}
	\caption{Algorithm searching for global events}
	\renewcommand{\algorithmicrequire}{\textbf{Input:}}
	\renewcommand{\algorithmicensure}{\textbf{Output:}}
	
	\begin{algorithmic}[1]
		\Require vector set $\mathcal{E}$
		\Ensure set of global events $S_{\Lambda}(\mathcal{E})$
		
		\State $S_{\Lambda}(\mathcal{E})=[]$		
		\State Generate the exclusivity $G$ graph corresponding to $E$.	
		\State Enumerate every normalized orthogonal basis of $E$ and add them into $\mathcal{M}(\mathcal{E})$.
		\State Make copy of graph $G$ and delete every vertices corresponding vectors in all $C_i\in\mathcal{M}(\mathcal{E})$ in the copy to get subgraph $G'$.
		
		\For{each $(v_1,...,v_k), v_i\in C_i, i\in\{1,2,...,k\}, k=|\mathcal{M}(\mathcal{E})|$}
		\If{$\{v_1,...,v_k\}$ is the independent set of $G$}
		\State Add $\{v_1,...,v_k\}$ into $S_{\Lambda}(\mathcal{E})$.
		\For{each $I$ is the independent set of $G'$}
		\State Add $\{v_1,...,v_k\}\cup I$ into $S_{\Lambda}(\mathcal{E})$.
		\EndFor
		\EndIf
		\EndFor
		
		\Return $S_{\Lambda}(\mathcal{E})$
		
	\end{algorithmic}
\end{algorithm}

\begin{longtable}{c|ccccccccccccc}
	\label{table:appendix2}
	\diagbox{Ge}{Me} & $v_{k_1}$ & $v_{k_2}$ & $v_{k_3}$ & $v_4$ & $v_5$ & $v_6$ & $v_7$ & $v_8$ & $v_9$ & $v_A$ & $v_B$ & $v_C$ & $v_D$ \\
	\hline
	$\lambda_1$ & 1 & 0 & 0 & 0 & 1 & 1 & 0 & 0 & 0 & 0 & 0 & 0 & 0  \\
	$\lambda_2$ & 1 & 0 & 0 & 0 & 1 & 1 & 0 & 0 & 0 & 1 & 0 & 0 & 0  \\
	$\lambda_3$ & 1 & 0 & 0 & 0 & 1 & 0 & 0 & 0 & 1 & 0 & 0 & 0 & 0  \\
	$\lambda_4$ & 1 & 0 & 0 & 0 & 1 & 0 & 0 & 0 & 1 & 0 & 0 & 1 & 0  \\
	$\lambda_5$ & 1 & 0 & 0 & 0 & 0 & 1 & 0 & 1 & 0 & 0 & 0 & 0 & 0  \\
	$\lambda_6$ & 1 & 0 & 0 & 0 & 0 & 1 & 0 & 1 & 0 & 0 & 1 & 0 & 0  \\
	$\lambda_7$ & 1 & 0 & 0 & 0 & 0 & 0 & 0 & 1 & 1 & 0 & 0 & 0 & 0  \\
	$\lambda_8$ & 1 & 0 & 0 & 0 & 0 & 0 & 0 & 1 & 1 & 0 & 0 & 0 & 1  \\
	$\lambda_9$ & 0 & 1& 0 & 1 & 0 & 1 & 0 & 0 & 0 & 0 & 0 & 0 & 0  \\
	$\lambda_{10}$ & 0 & 1 & 0 & 1 & 0 & 1 & 0 & 0 & 0 & 0 & 1 & 0 & 0  \\
	$\lambda_{11}$ & 0 & 1 & 0 & 1 & 0 & 0 & 0 & 0 & 1 & 0 & 0 & 0 & 0  \\
	$\lambda_{12}$ & 0 & 1 & 0 & 1 & 0 & 0 & 0 & 0 & 1 & 0 & 0 & 1 & 0  \\
	$\lambda_{13}$ & 0 & 1 & 0 & 0 & 0 & 1 & 1 & 0 & 0 & 0 & 0 & 0 & 0  \\
	$\lambda_{14}$ & 0 & 1 & 0 & 0 & 0 & 1 & 1 & 0 & 0 & 1 & 0 & 0 & 0  \\
	$\lambda_{15}$ & 0 & 1 & 0 & 0 & 0 & 0 & 1 & 0 & 1 & 0 & 0 & 0 & 0  \\
	$\lambda_{16}$ & 0 & 1 & 0 & 0 & 0 & 0 & 1 & 0 & 1 & 0 & 0 & 0 & 1  \\
	$\lambda_{17}$ & 0 & 0 & 1 & 1 & 1 & 0 & 0 & 0 & 0 & 0 & 0 & 0 & 0  \\
	$\lambda_{18}$ & 0 & 0 & 1 & 1 & 1 & 0 & 0 & 0 & 0 & 0 & 0 & 1 & 0  \\
	$\lambda_{19}$ & 0 & 0 & 1 & 1 & 0 & 0 & 0 & 1 & 0 & 0 & 0 & 0 & 0  \\
	$\lambda_{20}$ & 0 & 0 & 1 & 1 & 0 & 0 & 0 & 1 & 0 & 0 & 1 & 0 & 0  \\
	$\lambda_{21}$ & 0 & 0 & 1 & 0 & 1 & 0 & 1 & 0 & 0 & 0 & 0 & 0 & 0  \\
	$\lambda_{22}$ & 0 & 0 & 1 & 0 & 1 & 0 & 1 & 0 & 0 & 1 & 0 & 0 & 0  \\
	$\lambda_{23}$ & 0 & 0 & 1 & 0 & 0 & 0 & 1 & 1 & 0 & 0 & 0 & 0 & 0  \\
	$\lambda_{24}$ & 0 & 0 & 1 & 0 & 0 & 0 & 1 & 1 & 0 & 0 & 0 & 0 & 1 \\
	
\caption{Global events (KS-assignment) on Yu-Oh set, in which 'Ge' denotes Global event and 'Me' denotes Marginal events.}
\end{longtable}

\end{appendices}


\bibliography{sn-bibliography}

\end{document}